\documentclass{llncs}
\usepackage{amsmath,amssymb}
\usepackage[boxed,noline,longend]{algorithm2e}
\usepackage{xspace}
\usepackage{graphicx,subfigure}
\usepackage{xcolor}

\graphicspath{{../figures/}}

\title{Quasi-Polynomial Local Search for Restricted Max-Min Fair Allocation\thanks{This research was supported by ERC Advanced investigator grants 228021 and 226203.}}
\author{Lukas Polacek\inst{1} and Ola Svensson\inst{2}}
\institute{KTH Royal Institute of Technology, Sweden
  \email{polacek@csc.kth.se}
  \and EPFL, Switzerland
  \email{ola.svensson@epfl.ch}
}
\date{\today}

\newcommand{\res}{\ensuremath{\mathcal{R}}\xspace}
\newcommand{\players}{\ensuremath{\mathcal{P}}\xspace}
\newcommand{\conf}[1]{\ensuremath{\mathcal{C}(#1)}}
\newcommand{\hide}[1]{}

\begin{document}
\maketitle

\begin{abstract}
The restricted max-min fair allocation problem (also known as the
restricted Santa Claus problem) is one of few problems that enjoys the
intriguing status of having a better estimation algorithm than
approximation algorithm. Indeed, Asadpour et al.~\cite{AFS08} proved
that a certain configuration LP can be used to estimate the optimal
value within a factor ${1}/{(4+\epsilon)}$, for any $\epsilon>0$, but at
the same time it is not known how to efficiently find a solution with
a comparable performance guarantee.

A natural question that arises from their work is if the difference
between these guarantees is inherent or because of a lack of suitable
techniques. We address this problem by giving a quasi-polynomial
approximation algorithm with the mentioned performance guarantee. More
specifically, we modify the local search of~\cite{AFS08} and provide a
novel analysis that lets us significantly improve the bound on its
running time: from $2^{O(n)}$ to $n^{O(\log n)}$. Our techniques also
have the interesting property that although we use the rather complex
configuration LP in the analysis, we never actually solve it and
therefore the resulting algorithm is purely combinatorial.
\end{abstract}

\section{Introduction}


We consider the problem of indivisible resource allocation in the
following classical setting: a set \res of available resources shall
be allocated to a set \players of players where the value of a set of
resources for player $i$ is given by the function $f_i : 2^\res
\mapsto \mathbb{R}$.
This is a very general setting and dependent on the specific goals of
the allocator several different objective functions have been studied.

One natural objective, recently studied
in~\cite{DS06,Feige06,FV06,Vondrak08}, is to maximize the social
welfare, i.e., to find an allocation $\pi: \res \mapsto \players$ of
resources to players so as to maximize $\sum_{i\in \players} f_i(
\pi^{-1}(i))$.  However, this approach is not suitable in settings
where the property of ``fairness'' is desired.  Indeed, it is easy to
come up with examples where an allocation that maximizes the social
welfare assigns all resources to even a single player. In this
paper we address this issue by studying algorithms for finding
``fair'' allocations. More specifically, fairness is modeled by
evaluating an allocation with respect to the satisfaction of the least
happy player, i.e., we wish to find an allocation $\pi$ that maximizes
$\min_{i\in \players} f_i(\pi^{-1}(i))$. In contrast to maximizing the
social welfare, the problem of maximizing fairness is already
$\mathsf{NP}$-hard when players have linear value functions. In order
to simplify notation for such functions we denote $f_i({j})$ by
$v_{i,j}$ and hence we have that $f_i(\pi^{-1}(i)) = \sum_{j\in
  \pi^{-1}(i)} v_{i,j}$. This problem has recently received
considerable attention in the literature and is often referred to as
the \emph{max-min fair allocation} or the \emph{Santa Claus} problem.

One can observe that the max-min fair allocation problem is similar to
the classic problem of scheduling jobs on unrelated machines to
minimize the makespan, where we are given the same input but wish to
find an allocation that minimizes the maximum instead of one that
maximizes the minimum. In a classic paper~\cite{LST90}, Lenstra,
Shmoys \& Tardos gave a $2$-approximation algorithm for the scheduling
problem and proved that it is $\mathsf{NP}$-hard to approximate the
problem within a factor less than $1.5$. The key step of their
$2$-approximation algorithm is to show that a certain linear program,
often referred to as the assignment LP, yields an additive
approximation of $v_{\max} = \max_{i,j} v_{i,j}$.  Bez\'akov\'a and
Dani~\cite{BD05} later used these ideas for max-min fair allocation to
obtain an algorithm that always finds a solution of value at least
$OPT-v_{\max}$, where $OPT$ denotes the value of an optimal
solution. However, in contrast to the scheduling problem, this
algorithm and more generally the assignment LP gives no approximation
guarantee for max-min fair allocation in the challenging cases when
$v_{\max} \geq OPT$.

In order to overcome this obstacle, Bansal \& Sviridenko~\cite{BS06}
proposed a stronger linear program relaxation, known as the
configuration LP, for the max-min fair allocation problem. The
configuration LP that we describe in detail in Section~\ref{sec:CLP}
has been vital to the recent progress on better approximation
guarantees.  Asadpour \& Saberi~\cite{AS07} used it to obtain a
$\Omega(1/\sqrt{|\players|} (\log |\players|)^3)$-approximation
algorithm which was later improved by Bateni et al.~\cite{Bateni09}
and Chakrabarty et al.~\cite{Chakrabarty09} to algorithms that return
a solution of value at least $\Omega(OPT/|\players|^\epsilon)$ in time
$O(|\players|^{1/\epsilon})$.

The mentioned guarantee $\Omega(OPT/|\players|^\epsilon)$ is rather
surprising because the integrality gap of the configuration LP is no
better than $O(OPT/\sqrt{|\players|})$~\cite{BS06}. However, in
contrast to the general case, the configuration LP is significantly
stronger for the prominent special case where values are of the form
$v_{i,j} \in \{v_j, 0\}$. This case is known as the \emph{restricted}
max-min fair allocation or the restricted Santa Claus problem and is the focus of our
paper. The worst known integrality gap for the restricted case is
$1/2$ and it is known~\cite{BD05} that it is $\mathsf{NP}$-hard to beat
this factor (which is also the best known hardness result for the
general case).
Bansal \& Sviridenko~\cite{BS06} first used the configuration LP to
obtain an $O(\log\log \log |\players| / \log \log
|\players|)$-approximation algorithm for the restricted max-min fair
allocation problem. They also proved several structural properties
that were later used by Feige~\cite{Feige08} to prove that the
integrality gap of the configuration LP is in fact constant in the
restricted case. The proof is based on repeated use of Lov\'{a}sz
local lemma and was 
turned into a polynomial time algorithm~\cite{HSS10}.

The approximation guarantee obtained by combining~\cite{Feige08}
and~\cite{HSS10} is a large constant and is far away from the best
known analysis of the configuration LP by Asadpour et
al.~\cite{AFS08}.
More specifically, they proved in~\cite{AFS08} that the integrality
gap is lower bounded by $1/4$ by designing a beautiful local search
algorithm that eventually finds a solution with the mentioned
approximation guarantee, but is only known to converge in exponential
time.  As the configuration LP can be solved up to any precision in
polynomial time, this means that we can approximate the value of an
optimal solution within a factor $1/(4+\epsilon)$ for any $\epsilon
>0$ but it is not known how to efficiently find a solution with a
comparable performance guarantee. Few other problems enjoy this
intriguing status (see e.g. the overview article by
Feige~\cite{FeigeSurv08}). One of them is the restricted assignment
problem\footnote{Also here the restricted version of the  problem is the special case where $v_{ij} \in \{v_j, \infty\}$ ($\infty$ instead of $0$ since we are minimizing).},
  for which the second author in \cite{SME11} developed the techniques
from~\cite{AFS08} to show that the configuration LP can be used to
approximate the optimal makespan within a factor $33/17 + \epsilon$
improving upon the $2$-approximation by Lenstra, Shmoys \&
Tardos~\cite{LST90}. Again it is not known how to efficiently find a
schedule of the mentioned approximation guarantee. However, these
results indicate that an improved understanding of the configuration
LP is likely to lead to improved approximation algorithms for these
fundamental allocation problems.


In this paper we make progress that further substantiates this
point. We modify the local search of~\cite{AFS08} and present
a novel analysis that allows us to significantly improve the bound on the
running time from an exponential guarantee to a quasi-polynomial
guarantee.
\begin{theorem}
\label{thm:main}
  For any $\epsilon \in (0,1]$, we can find a
  $\frac{1}{4+\epsilon}$-approximate solution to  restricted
  max-min fair allocation in time $n^{O\left(\frac{1}{\epsilon}
    \log n\right)}$, where $n= |\players| + |\res|$.
\end{theorem}

In Section~\ref{sec:algodesc}, we give an overview of the local search
of~\cite{AFS08} together with our modifications. The main modification
is that at each point of the local search, we carefully select which step
to take in the case of several options, whereas in the original
description~\cite{AFS08} an arbitrary choice was made. We then use
this more stringent description with a novel analysis (Section~\ref{sec:algoanal})
that uses the dual of the configuration LP as in~\cite{SME11}. The
main advantage of our analysis (of the modified local search) is that
it allows us to obtain a better upper bound on the search space of the
local search and therefore also a better bound on the run-time.
Furthermore, our techniques have the interesting property that
although we use the rather complex configuration LP in the analysis,
we never actually solve it.  This gives hope to the interesting
possibility of a polynomial time algorithm that is purely
combinatorial and efficient to implement (in contrast to solving the
configuration LP) with a good approximation ratio.

Finally, we note that our approach currently has a similar dependence on
$\epsilon$ as in the case of solving the configuration LP since, as
mentioned above, the linear program itself can only be solved
approximately. However, our hidden constants are small and for a
moderate $\epsilon$ we  expect that our combinatorial approach is
already more attractive than solving the configuration LP.

\section{The Configuration LP}
\label{sec:CLP}
The intuition of the configuration linear program (LP) is that any
allocation of value $T$ needs to allocate a bundle or configuration
$C$ of resources to each player $i$ so that $f_i(C) \geq T$.  Let
$\conf{i,T}$ be the set of those configurations that have value at
least $T$ for player $i$. In other words, $\conf{i,T}$ contains all
those subsets of resources that are feasible to allocate to player $i$
in an allocation of value $T$. For a guessed value of $T$, the
configuration LP therefore has a decision variable $x_{i, C}$ for each
player $i\in \players$ and configuration $C \in \conf{i,T}$ with the
intuition that this variable should take value one if and only if the
corresponding set of resources is allocated to that player. The configuration LP
$CLP(T)$ is a feasibility program and it is defined as follows:
\begin{equation*}
\boxed{%
        \begin{minipage}{10cm}%
          \begin{align*}
            \sum_{C\in \conf{i,T}} x_{i,C} &\geq 1  &  i \in \players \\[1mm]
            \sum_{i,C: j\in C, C\in \conf{i,T}} x_{i,C} &\leq 1 & j \in \res \\[1mm]
            x & \geq 0
          \end{align*}
          \vspace{-0.4cm}
        \end{minipage}%
}
\end{equation*}
The first set of constraints ensures that each player should receive
at least one bundle and the second set of constraints ensures that a
resource is assigned to at most one player.

If $CLP(T_0)$ for some $T_0$ is feasible, then $CLP(T)$ is also feasible for all
$T\leq T_0$, because $\conf{i, T_0}\subseteq \conf{i, T}$ and thus a solution to
$CLP(T_0)$ is a solution to $CLT(T)$ as well. Let $T_{OPT}$
be the maximum of all such values. Every
feasible allocation is a feasible solution of configuration LP, hence
$T_{OPT}$ is an upper bound on the value of the optimal allocation.

We note that the LP has exponentially many constraints; however, it is
known that one can approximately solve it to any desired
accuracy by designing a polynomial time (approximate) separation algorithm for the
dual~\cite{BS06}. Although our approach does not require us to solve the linear
program, the dual shall play an important role in our analysis.  By
associating a variable $y_i$ with each constraint in the first set of
constraints, a variable $z_j$ with each constraint in the second set of
constraints, and letting the primal have the objective function of
minimizing the zero function, we obtain the dual program:
\begin{equation*}
\boxed{%
        \begin{minipage}{10cm}%
          \begin{align*}
            \max & \mbox{  } \sum_{i\in \players} y_i - \sum_{j\in \res} z_j  && \\[2mm]
            y_i &\leq             \sum_{j\in C} z_j &  i\in \players, C \in
            \conf{i, T}\\[1mm]
            y,z & \geq 0
          \end{align*}
          \vspace{-0.4cm}
        \end{minipage}%
}
\end{equation*}
\section{Local Search with Better Run-time Analysis}

In this section we modify the algorithm by Asadpour et
al.~\cite{AFS08} in order to significantly improve the run-time
analysis: we obtain a $1/(4+\epsilon)$-approximate solution in
run-time bounded by $n^{O(1/\epsilon \log n)}$ whereas the original
local search is only known to converge in time $2^{O(n)}$. For better
comparison, we can write $n^{O(1/\epsilon \log n)} = 2^{O(1/\epsilon
  \log^2 n)}$. Moreover, our modification has the nice side effect
that we actually never solve the complex configuration LP --- we only
use it in the analysis.

\subsection{Description of Algorithm}
\label{sec:algodesc}
Throughout this section we assume that $T$ --- the guessed optimal value --- is such that $CLP(T)$ is feasible. We
shall find an $1/\alpha$ approximation where $\alpha$ is a parameter such that
$\alpha>4$. As we will see, the selection of $\alpha$ has the following
trade-off: the closer $\alpha$ is to $4$ the worse bound on the run-time we get.

We note that if $CLP(T)$ is not feasible and thus $T$ is more than $T_{OPT}$,
our algorithm makes no guarantees. It might fail to find an allocation, which
means that $T>T_{OPT}$. We can use this for a standard binary search  on the
interval $[0, \frac{1}{|\players|}\sum_i v_i]$ so that in the end we find an
allocation with a value at least $T_{OPT}/\alpha$.

\subsubsection{Max-min fair allocation is a bipartite hypergraph problem.}
Similar to~\cite{AFS08}, we view the max-min fair allocation problem as a
matching problem in the bipartite hypergraph $G=(\players, \res, E)$. Graph $G$
has an hyperedge $\{i\} \cup C$ for each player $i\in \players$ and
configuration $C \subseteq \res$ that is feasible with respect to the desired
approximation ratio $1/\alpha$, i.e., $f_i(C) \geq T/\alpha$, and minimal in the
sense that $f_i(C') < T/\alpha$ for all $C' \subset C$. Note that the graph
might have exponentially many edges and the algorithm therefore never keeps
an explicit representation of all edges.

From the construction of the graph it is clear that a matching covering all
players corresponds to a solution with value at least $T/\alpha$. Indeed, given
such a matching $M$ in this graph, we can assign matched resources to the
players and everyone gets resources with total value of at least $T/\alpha$.

\subsubsection{Alternating tree of ``add'' and ``block'' edges.}

The algorithm of Asadpour et al.~\cite{AFS08} can be viewed as
follows. In the beginning we start with an empty matching and then we
increase its size in every iteration by one, until all players are
matched. In every iteration we build an alternating tree rooted in a
currently unmatched player $p_0$ in the attempt to find an alternating
path to extend our current matching $M$. The alternating tree has two
types of edges: edges in the set $A$ that we wish to \emph{add} to the
matching and edges in the set $B$ that are currently in the matching
but intersect edges in $A$ and therefore \emph{block} them from being
added to the matching. While we are building the alternating tree to
find an alternating path, it is important to be careful in the
selection of edges, so as to guarantee eventual termination. As
in~\cite{AFS08}, we therefore define the concept of addable and
blocking edges.

Before giving these definitions, it will be convenient to introduce
the following notation. For a set of edges $F$, we denote by $F_\res$
all resources contained in edges in $F$ and similarly $F_\players$
denotes all players contained in edges in $F$. We also write $e_\res$
instead of $\{e\}_\res$ for an edge $e$ and use $e_\players$ to denote
the player in $e$.

\begin{definition}
We call an edge $e$ addable if $e_\res\cap (A_\res\cup B_\res) = \emptyset$ and
$e_\players\in \{p_0\}\cup A_\players\cup B_\players$.
\end{definition}
\begin{definition}
An edge $b$ in the matching $M$ is blocking $e$ if
$e_\res\cap b_\res \not= \emptyset$.
\end{definition}
Note that an addable edge matches a player in the tree with resources
that currently do not belong to any edge in the tree and that the
edges blocking an edge $e$ are exactly those in the matching that
prevent us from adding $e$. For a more intuitive understanding of
these concepts see Figure \ref{fig:alter-tree} in Section \ref{sec:example}.

The idea of building an alternating tree is similar to standard matching
algorithms using augmenting paths. However, one key difference is that the matching can be
extended once an alternating path is found in the graph case, whereas
the situation is more complex in the considered hypergraph case, since
a single hyperedge might overlap several hyperedges in the matching. It is due
to this complexity that it is more difficult to bound the running time
of the hypergraph matching algorithm of~\cite{AFS08} and our improved
running time is obtained by analyzing a modified version where we
carefully select in which order the edges should be added to the
alternating tree and drop edges from the tree beyond certain distance.

We divide resources into 2 groups. \emph{Fat resources} have value at least
$T/\alpha$ and \emph{thin resources} have less than $T/\alpha$. Thus any edge
containing a fat resource contains only one resource and is called \emph{fat
edge}. Edges containing thin resources are called \emph{thin edges}. Our
algorithm always selects an addable edge of minimum distance to the root $p_0$
according to the following convention. The length of a thin edge in the tree is
one and the length of a fat edge in the tree is zero. Edges not in the tree have
infinite length. Hence, the \emph{distance of a vertex} from the root is the
number of thin edges between the vertex and the root and, similarly, the
\emph{distance of an edge $e$} is the number of thin edges on the way to $e$
from $p_0$ including $e$ itself. We also need to refer to distance of an
addable edge that is not yet in the tree. In that case we take the distance as
if it was in the tree. Finally, by the \emph{height of the alternating tree} we
refer to the maximum distance of a resource from the root.

\subsubsection{Algorithm for extending a partial matching.}
\begin{algorithm}[h!]
\DontPrintSemicolon
\caption{Increase the size of the matching}
\label{increase-match}
\SetKwInOut{Input}{Input}\SetKwInOut{Output}{Output}
\Input{A partial matching $M$}
\Output{A matching of increased size assuming that $T$ is at most $T_{OPT}$}
\BlankLine
 Find an unmatched player $p_0\in \players$, make it a root of the alternating tree\;
\While{there is an addable edge within distance $2\log_{(\alpha-1)/3}(|\players|)+1$}{

  \Indp
  Find an addable edge $e$ of minimum distance from the root\;
  $A\leftarrow A\cup \{e\}$\;
  \eIf {$e$ has blocking edges $b_1, \dots, b_k$} {
    \Indp
    $B\leftarrow B\cup \{b_1, \dots, b_k\}$\;
    \Indm
  }(\tcp*[h]{collapse procedure})
  {
    \Indp
    \While{$e$ has no blocking edges} {
      \Indp
      \eIf {there is an edge $e'\in B$ such that $e'_\players=e_\players$} {
        \Indp
        $M \leftarrow M\setminus \{e'\}\cup \{e\}$\;
        $A\leftarrow A\setminus \{e\}$\;
        $B \leftarrow B\setminus \{e'\}$\;
        Let $e''\in A$ be the edge that $e'$ was blocking\;
        $e\leftarrow e''$\;
        \Indm
      }
      {
        \Indp
        $M\leftarrow M\cup\{e\}$\;
        \Return $M$\;
        \Indm
      }
    \Indm
    }
    Let $e'$ be the blocking edge that was last removed from $B$\;
    Remove all edges in $A$ of greater distance than $e'$ and the edges in $B$ that blocked these edges\;
    \Indm
  }
}
\Return $T_{OPT}$ is less than $T$
\end{algorithm}

Algorithm \ref{increase-match} summarizes the modified procedure for increasing
the size of a given matching by also matching a previously unmatched player
$p_0$. For better understanding of the algorithm, we included an example of an
algorithm execution in Figure \ref{fig:alter-tree} in Section
\ref{sec:example}.

Note that the algorithm iteratively tries to find addable edges of
minimum distance to the root. On the one hand, if the picked edge $e$
has blocking edges that prevents it from being added to the matching,
then the blocking edges are added to the alternating tree and the
algorithm repeatedly tries to find addable edges so as to make
progress by removing the blocking edges.

On the other hand, if edge $e$ has no blocking edges, then this means that the set of
resources $e_\res$ is free, so we make progress by adding $e$ to the matching $M$. If the
player was not previously matched, it is the root $p_0$ and we
increased the size of the matching. Otherwise the player $e_\players$
was previously matched by an edge $e'\in B$ such that
$e'_\players=e_\players$, so we remove $e'$ from $M$ and thus
it is not a blocker anymore and can be removed
from $B$. This removal has decreased the number of blockers for an
edge $e''\in A$. If $e''$ has $0$ blockers, we recurse and repeat the same
procedure as with $e$. Note that this situation can be seen on Figure
\ref{fig:alt-2} and \ref{fig:alt-3} in Section \ref{sec:example}.

\subsection{Example of Algorithm Execution}
\label{sec:example}
\begin{figure}[h]
  \begin{center}
    \begin{tabular}{| c | c |}
      \hline
      \subfigure[Step 1]{
        \includegraphics[page=1, width=0.39\textwidth]{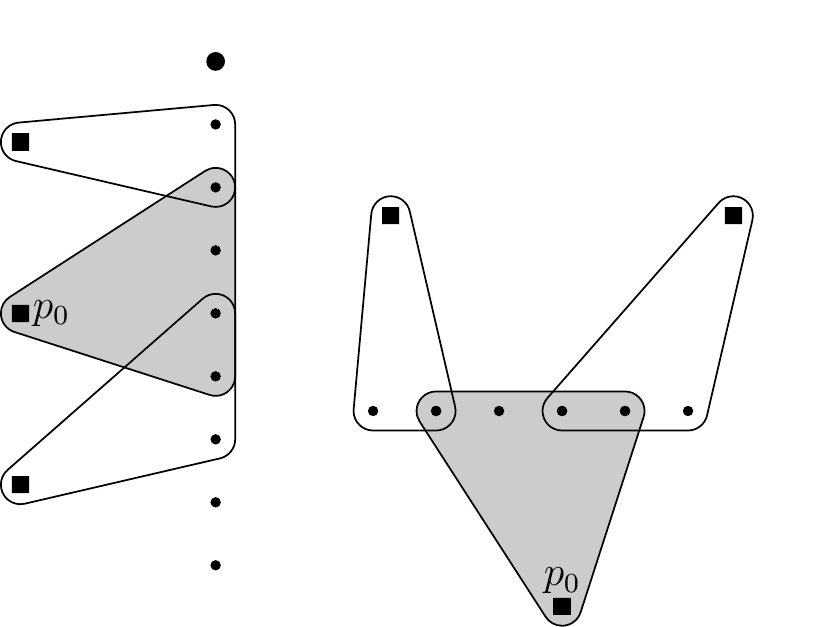}
        \label{fig:alt-1}
      }&
      \subfigure[Step 2]{
        \includegraphics[page=2, width=0.39\textwidth]{matching-example-1}
        \label{fig:alt-2}
      }\\
      \hline
      \subfigure[Step 3]{
        \includegraphics[page=3, width=0.39\textwidth]{matching-example-1}
        \label{fig:alt-3}
      }&
      \subfigure[Step 4]{
        \includegraphics[page=4, width=0.39\textwidth]{matching-example-1}
        \label{fig:alt-4}
      }\\
      \hline
    \end{tabular}
  \end{center}
  \caption{Alternating tree visualization. The right part of every picture
  is the alternating tree and to the left we display the positions of edges in
  the tree in the bipartite graph. Gray edges are in the set $A$ and white edges
  are in the set $B$.}
  \label{fig:alter-tree}
\end{figure}

Figure \ref{fig:alter-tree} is a visualization of an execution of Algorithm
\ref{increase-match}. The right part of every picture is the alternating tree
and to the left we display the positions of edges in the tree in the bipartite
graph.  Gray edges are $A$-edges and white are $B$-edges.

In Figure \ref{fig:alt-1} we start by adding an $A$-edge to the tree. There are
2 edges in the matching intersecting this edge, so we add them as blocking
edges. Then in Figure \ref{fig:alt-2} we add a fat edge that has no blockers, so
we add it to the matching and thus remove one blocking edge, as we can see in
Figure \ref{fig:alt-3}. Then in Figure \ref{fig:alt-4} we add a thin edge which
has no blockers. Now the $A$ and $B$ edges form an alternating path, so by
swapping them we increase the size of the matching and the algorithm terminates.

Note that the fat edge in step 2 is added before the thin edge from step 4,
because it has shorter distance from the root $p_0$. Recall that the distance of
an edge $e$ is the number of thin edges between $e$ and the root including $e$,
thus the distance of the fat edge is 2 and the distance of the thin edge is 3.

\subsection{Analysis of Algorithm}

\label{sec:algoanal}

Let the parameter $\alpha$ of the algorithm equal $4+\epsilon$ for some
$\epsilon\in (0, 1]$. We first prove that Algorithm~\ref{increase-match} terminates in
time $n^{O\left(\frac{1}{\epsilon}\log n\right)}$ where $n=|\players| + |\res|$
and, in the following subsection, we show that it returns a matching of
increased size if $CLP(T)$ is feasible.

Theorem~\ref{thm:main}~then follows from that, for each guessed value of $T$,
Algorithm~\ref{increase-match} is at most invoked $n$ times and we can find the
maximum value $T$ for which our algorithm finds an allocation by binary search
on the interval $[0, \frac{1}{|\players|}\sum_i v_i]$. Since we can assume that the numbers in the
input have bounded precision, the binary search only adds a polynomial factor to
the running time.

\subsubsection{Run-time Analysis.}
We bound the running time of Algorithm~\ref{increase-match} using that
the alternating tree has height at most
$O(\log_{(\alpha-1)/3}|\players|) = O\left(\frac{1}{\epsilon} \log
|\players|\right)$. The proof is similar to the termination proof
in~\cite{AFS08} in the sense that we associate a signature vector with
each tree and then show that its lexicographic value
decreases. However, one key difference is that instead of associating a
value with \emph{each edge} of type $A$ in the tree, we associate a
value with each ``layer'' that consists of \emph{all edges} of a certain
distance from the root. This allows us to directly associate the
run-time with the height of the alternating tree.

When considering an alternating tree it is convenient to partition $A$ and $B$
into $A_0, A_1, \dots, A_{2k}$ and $B_0, B_1, \dots, B_{2k}$ respectively by the
distance from the root, where $2k$ is the maximum distance of an edge in the
alternating tree (it is always an even number). Note that $B_i$ is empty for all
odd $i$. Also, $A_{2i}$ contains only fat edges and $A_{2i+1}$ only thin edges.
For a set of edges $F$ we denote by $F^t$ all the thin edges in $F$ and by $F^f$
all the fat edges in $F$.  For a set of edges $F$ denote by $F^t$ all the thin
edges in $F$ and by $F^f$ all the fat edges in $F$. We also use $\res^t$ to
denote thin resources and $\res^f$ to denote fat resources.

\begin{lemma}
\label{lemma:termination}
For a desired approximation guarantee of $1/\alpha = 1/(4+\epsilon)$, Algorithm~\ref{increase-match} terminates in time
$n^{O\left(\frac{1}{\epsilon}\log n\right)}$.
\end{lemma}
\begin{proof}
We analyze the run-time of Algorithm~\ref{increase-match} by
associating a signature vector with the alternating tree of each
iteration. The signature vector of an alternating tree is then defined to be
\begin{align*}
 ( -|A^f_0|, |B^f_0|, &-|A^t_1|, |B^t_2|,  -|A^f_2|, |B^f_2|, \\
&-|A^t_3|, |B^t_4|,  -|A^f_4|, |B^f_4|, \\
&  \vdots \\
& -|A^t_{2k-1}|, |B^t_{2k}|,  -|A^f_{2k}|, |B^f_{2k}|, \infty).
\end{align*}

We prove that each addition of an edge decreases the lexicographic value of
the signature or increases the size of the matching.

On the one hand, if we add an edge with no blocking edges, we either completely
collapse the alternating tree or collapse only a part of it.
If we completely collapse the tree then the algorithm terminates. Otherwise, let
$e'$ be the last blocking edge that was removed from $B$ by the algorithm during
the collapse procedure. Also let $B'$ and $A'$ be the sets of  blocking edges
and addable edges obtained after the collapse procedure. Note that $e'$ is a
thin edge because otherwise $e'$ was blocking a fat edge $e$ that after the
removal of $e'$ had no more blocking edges which in turn contradicts that $e'$
was the last blocking edge removed from $B$. Let $2\ell$ be the distance of
$e'$, i.e., $e' \in B^t_{2\ell}$. As the algorithm drops all edges in $A$ of
distance at least $2 \ell+1$ and all edges in $B$ that  blocked these edges, the
partial collapse of the tree changes its 
signature to $$( -|A'^f_0|, |B'^f_0|, -|A'^t_1|, |B'^t_2|,  -|A'^f_2|, |B'^f_2|, \dots, -|A'^t_{2\ell-1}|,
|B'^t_{2\ell}|,  -|A'^f_{2\ell}|, |B'^f_{2\ell}|, \infty),$$ which equals
$$( -|A^f_0|, |B^f_0|, -|A^t_1|, |B^t_2|,  -|A^f_2|, |B^f_2|, \dots, -|A^t_{2\ell-1}|,
|B^t_{2\ell}|-1,  -|A'^f_{2\ell}|, |B'^f_{2\ell}|, \infty).$$
Thus we either
increase the size of the matching or decrease the signature of the alternating
tree.

On the other hand, if the added edge $e$ has blocking edges, there are two
cases. We either open new layers $A_{2k+1}=\{e\}$ and $B_{2k+2}$ where $e$ is a
thin edge and the signature gets smaller, since $-|A^t_{2k+1}|<\infty$. If we do
not open a new layer, we increase the size of some $A_\ell$ by either a thin or fat edge and
$-(|A_\ell|+1)<-|A_\ell|$, so in this case the signature decreases too.

The algorithm only runs as long as the height of the alternating tree is at most
$O(\log_{(\alpha-1)/3}|\players|)=O(\log_{1+\epsilon/3} |\players|)$. This can
be rewritten as
$
  O\left(\frac{\log |\players|}{\log (1+\epsilon/3)}\right)=
  O\left(\frac{\log |\players|}{\epsilon}\right)
$
where the equality follows from $x\leq 2\log(1+x)$ for $x\in(0,1]$ and we only
consider $\epsilon\in(0, 1]$. There are at most
$|\players|$ possible values for each position in a signature, so the total
number of signatures encountered during the execution of
Algorithm~\ref{increase-match} is $|\players|^{O\left(\frac{1}{\epsilon} \log
|\players|\right)}$. As adding an edge happens in polynomial time in $n =
|\players| + |\res|$, we conclude that Algorithm~\ref{increase-match} terminates
in time $n^{O\left(\frac{1}{\epsilon}\log n\right)}$. \qed
\end{proof}

\subsubsection{Correctness of Algorithm~\ref{increase-match}.}
\label{sec:correctness}
We show that Algorithm~\ref{increase-match} is correct, i.e., that it
returns an increased matching if $CLP(T)$ is feasible.

We have already proved that the algorithm terminates in
Lemma~\ref{lemma:termination}. The statement therefore follows from
proving that the condition of the while loop always is satisfied
assuming that the configuration LP is feasible. In other words, we
will prove that there always is an addable edge within the required
distance from the root. This strengthens the analogous statement
of~\cite{AFS08} that states that there always is an addable edge (but
without the restriction on the search space that is crucial for our run-time analysis).  We shall do so by proving
that the number of thin blocking edges increases quickly with respect
to the height of the alternating tree and, as there cannot be more
than $|\players|$ blocking edges, this in turn bounds the height of
the tree.

We are now ready to state the key insight behind the analysis that shows that
the number of blocking edges increases as a function of $\alpha$ and the height
of the alternating tree.
\begin{lemma}
Let $\alpha >4$. Assuming that $CLP(T)$ is feasible, if there is no
addable edge $e$ within distance $2D+1$ from the root for some integer $D$, then
\begin{equation*}
\frac{\alpha-4}{3} \sum_{i=1}^D |B^t_{2i}| < |B^t_{2D+2}|.
\end{equation*}
\label{lem:combinatorial}
\end{lemma}

Before giving the proof of Lemma \ref{lem:combinatorial}, let us see how it
implies that there always is an addable edge within distance $2
\log_{(\alpha-1)/3}(|\players|)+1$ from the root assuming the configuration LP
is feasible, which in turn implies the correctness of
Algorithm~\ref{increase-match}.

\begin{corollary}
  \label{cor:height}
  If $\alpha>4$ and $CLP(T)$ is feasible, then there is always an addable edge
  within distance $2D+1$ from the root, where $D =
  \log_{(\alpha-1)/3}|\players|$.
\end{corollary}
\begin{proof}
The proof of the corollary follows intuitively from that Lemma~\ref{lem:combinatorial} says that
the number of blocking edges increases  exponentially in terms of the height of the tree and
therefore, as there are at most $|\players|$ blocking edges, the height must be $O_\alpha(
\log|\players|)$. We now proceed with the formal proof.

Let us first consider the case when $|B_2^t|=0$, i.e., there are no
thin edges in the alternating tree, so its height is $0$. Then there must be an
addable edge (of distance at most $1$), since otherwise, by the above lemma, we
get a contradiction $0=(\alpha-4)/3|B_2^t|<|B_4^t|=0$.

From now on assume that $|B_2^t|\geq 1$ and suppose toward contradiction that
there is no addable edge within distance $2D+1$. Let
\begin{equation*}
  \text{$b_i=\sum_{j=1}^i |B_{2j}^t|$ and $q=(\alpha-4)/3$.}
\end{equation*}
By Lemma~\ref{lem:combinatorial},
\begin{equation*}
  \text{$q b_i<|B^t_{2i+2}|$ for $i\leq D$ and $b_{i+1}=b_i+|B^t_{2i+2}|$,}
\end{equation*}
so $b_{i+1}>(1+q)b_i$ for all $i\leq D$, which in turn implies
\begin{equation*}
b_{D+1} = \sum_{j = 1}^{D+1} |B^t_{2j}| >(1+q)^{D}b_1  \geq (1+q)^D = |\players|,
\end{equation*}
where the last equality follows by the selection of $D$.  However, this is a
contradiction since the number of blocking edges and hence $b_{D+1}$ is at
most the number of players $|\players|$.
\qed

\end{proof}
We complete the correctness analysis of the algorithm by presenting the  proof of the key lemma.

\begin{proof}[Lemma~\ref{lem:combinatorial}]
Let $H_{2D+1}$ be the tree formed from the original alternating tree by taking
all edges of distance at most $2D+1$ plus edges in the set $B^t_{2D+2}$. The
following invariant holds throughout the execution of Algorithm
\ref{increase-match} and plays an important role in the analysis:
If there is an addable edge $e$ with respect to $H_{2D+1}$ within distance
$2D+1$, then $e$ is an addable edge within distance $2D+1$ with respect to the
original tree. Hence, in the proof of this lemma we only need to consider edges
in $H_{2D+1}$. The invariant trivially holds in the beginning of the algorithm
and is preserved when adding an edge with blockers, because an edge of minimum
distance is selected. The situation is more complex when an edge has no
blockers. Dropping off the edges beyond certain distance in Algorithm
\ref{increase-match} ensures that the invariant remains true even in this case.

Suppose toward contradiction that there is no addable edge within
distance $2D+1$ and
\begin{equation*}
  \frac{\alpha-4}{3} \sum_{i=1}^D |B^t_{2i}| \geq |B^t_{2D+2}|.
\end{equation*}
We show that this implies that the dual of the configuration LP is
unbounded, which in turn contradicts the assumption that the primal is feasible.
Recall that the objective function of the dual is $\max \sum_{i\in \players} y_i
- \sum_{j\in \res} z_j$.  Furthermore, as each solution $(y,z)$ of the dual can
be scaled by a scalar $c$ to obtain a new solution $(c\cdot y, c \cdot z)$, any
solution with positive objective implies unboundedness.

We proceed by defining such solution $(y^*, z^*)$, that is determined by the alternating tree.
More precisely, we take
\begin{equation*}
   y^*_i = \begin{cases}
     \frac{\alpha-1}{\alpha}  & \mbox{if $i\in\players$ is within distance $2D$
     from the root},\\
     0 & \mbox{otherwise,}
   \end{cases}
\end{equation*}
and
\begin{equation*}
  z_j^* = \begin{cases}
    (\alpha-1)/\alpha  & \mbox{if $j\in\res$ is fat and within distance $2D$
    from the root,}\\
    v_j/T & \mbox{if $j\in\res$ is thin and within distance $2D+2$ from the
    root,} \\
    0 & \mbox{otherwise.}
  \end{cases}
\end{equation*}

Let us first verify that $(y^*, z^*)$ is indeed a feasible solution. We have
chosen all $y_i, z_j$ to be non-negative, so it only remains to check the first
condition of the dual. Let $i\in \players$ and let $C$ be such that $f_i(C)\geq
T$, i.e., $C\in \conf{i, T}$. We distinguish between the two cases when $y_i =
0$ and $y_i = (\alpha-1)/\alpha$. On the one hand, if $y_i = 0$ we have that
$y_i \leq \sum_{j\in C}z_j$, since $\sum_{j \in C} z_j$ is always non-negative.

On the other hand, if $y_i=(\alpha-1)/\alpha$, then we have two sub-cases.
Either there is $z_j=(\alpha-1)/\alpha$ for some $j\in C$ and we have
$\sum_{j\in C} z_j\geq y_i$.  Otherwise $\sum_{j\in C} z_j=\sum_{j\in C\cap F}
v_j/T\geq\sum_{j\in C\cap F} v_{i,j}/T$, where $F$ is the set of resources which
are assigned positive value $z_j$. Suppose $\sum_{j\in C\cap F}
v_j/T<(\alpha-1)/\alpha$, then there is a set $R=C\setminus F\subseteq \res$
with $f_i(R)\geq T/\alpha$ and thus $\{i\}\cup R$ is an addable edge in
$H_{2D+1}$ and hence an addable edge within distance $2D+1$ in the original tree.
This contradicts the assumption that no such addable edges
exist, so $\sum_{j\in C} z_j\ge y_i$.

Having proved that $(y^*,z^*)$ is a feasible solution, the proof is
now completed by showing that the value of the solution is positive.
We have
\begin{equation*}
\sum_{i\in \players}y_i=\frac{\alpha-1}{\alpha} \left(1+\sum_{i=0}^D
|B_{2i}|\right),
\end{equation*}
since each player in the alternating tree has its unique blocking edge leading
to it except the root. For fat resources we have
\begin{equation*}
\sum_{j\in \res^f}z_j\leq\frac{\alpha-1}{\alpha}\sum_{i=0}^D|B_{2i}^f|,
\end{equation*}
since every fat edge contains only one fat resource by minimality.

For thin resources,
\begin{equation*}
\sum_{j\in \res^t}z_j\leq \frac{2}{\alpha} \sum_{i=1}^{D+1}
|A^t_{2i-1}| + \frac{1}{\alpha} \sum_{i=1}^{D+1} |B^t_{2i}|,
\end{equation*}
since the size of each thin edge is at most $2T/\alpha$ and the part of each
blocking edge not contained in any other $A$-edge is at most of size $T/\alpha$,
because otherwise the set of resources in the blocking edge would not be a
minimal set.

We also have $|A^t_{2i-1}|\leq |B^t_{2i}|$ for any $i$, since each adding edge has to
have at least one blocking edge.
This implies
\begin{equation*}
  \sum_{j\in \res}z_j\leq
  \frac{\alpha-1}{\alpha} \sum_{i=0}^D |B_{2i}^f|
  + \frac{3}{\alpha} \sum_{i=1}^{D+1} |B^t_{2i}|.
\end{equation*}
By the assumption toward contradiction,
\begin{equation*}
  \text{$|B^t_{2D+2}|\leq \frac{\alpha-4}{3} \sum_{i=1}^D |B^t_{2i}|$, so
  $3\sum_{i=1}^{D+1} |B^t_{2i}| \leq (\alpha-1)\sum_{i=1}^D |B^t_{2i}|$.}
\end{equation*}
This implies
\begin{equation*}
  \sum_{j\in \res}z_j\leq
  \frac{\alpha-1}{\alpha} \sum_{i=0}^D|B_{2i}^f|
  + \frac{\alpha-1}{\alpha} \sum_{i=1}^D|B^t_{2i}| <
  \frac{\alpha-1}{\alpha} \left(1+\sum_{i=0}^D |B_{2i}|\right) =
  \sum_{i\in \players} y_i,
\end{equation*}
so the dual is unbounded and we get a contradiction. \qed
\end{proof}

\section{Conclusions}
Asadpour et al. \cite{AFS08} raised as an open question whether their local
search (or a modified variant) can be shown to run in polynomial time. We made
progress toward proving this statement by showing that a modified local search
procedure finds a solution in quasi-polynomial time. Moreover, based on our
findings, we conjecture the stronger statement that there is a local search
algorithm that does not use the LP solution, i.e., it is combinatorial, and it
finds a $1/(4+\epsilon)$-approximate solution in polynomial time for any fixed
$\epsilon >0$.

\section{Acknowledgements}

We are grateful to Ji\v{r}\'{i} Sgall and Martin B\"{o}hm for pointing out a mistake in the description of the
algorithm and in the signature vector used in the runtime analysis  of an earlier version
of this paper.

\bibliographystyle{splncs03}
\bibliography{santa}

\end{document}